\documentclass[11pt]{article}

\usepackage[text={6.8in,8.5in},centering]{geometry}
\usepackage{times}
\usepackage{amsfonts}
\usepackage{amsmath}
\usepackage{amssymb}
\usepackage{amsthm}
\usepackage{graphicx}
\usepackage{subfig}
\usepackage{url}
\usepackage{hyperref}

\def\final{0} 
\ifnum\final=1  
\newcommand{\vnote}[1]{[{\small Vicky: \bf #1}]\marginpar{*}}
\newcommand{\sidecomment}[1]{\marginpar{\tiny #1}}
\else 
\newcommand{\vnote}[1]{}
\newcommand{\sidecomment}[1]{}
\fi  

\newtheorem{lemma}{Lemma}[section]
\newtheorem{theorem}[lemma]{Theorem}
\newtheorem{definition}[lemma]{Definition}

\newcommand{\ms}[1]{\ensuremath{\mathsf{#1}}}

\newcommand{\ket}[1]{\ensuremath{|#1\rangle}}
\newcommand{\argmax}{\operatornamewithlimits{arg\ max}}
\newcommand{\argmin}{\operatornamewithlimits{arg\ min}}

\def\reals{{\mathbb R}}

\newcommand{\ver}{{\ms{V}}}
\newcommand{\edge}{{\ms{E}}}

\newcommand{\Oedge}{{\ms{OE}}}

\newcommand{\uni}{U}

\renewcommand{\deg}{{\ms{deg}}}

\newcommand{\union}{\cup}
\newcommand{\nbr}{{\ms{nbr}}}

\newcommand{\onbr}{{\ms{Onbr}}}
\newcommand{\fnbr}{{\ms{Fnbr}}}

\newcommand{\energy}{{\mathcal{E}}}
\newcommand{\oy}{{\mathcal{Y}}}

\newcommand{\ham}{{\mathcal{H}}}

\newcommand{\mdef}{\stackrel{\mathrm{def}}{=}}

\newcommand{\emb}{{\ms{emb}}}

\newcommand{\Gemb}{{G_{\ms{emb}}}}

\begin{document}

\title{Minor-Embedding in Adiabatic Quantum Computation: I. The Parameter
  Setting Problem
}

\author{Vicky Choi 
\\{\em vchoi@dwavesys.com}
\\ D-Wave Systems Inc. }

\maketitle

\begin{abstract}
  We show that the NP-hard  quadratic
  unconstrained binary
  optimization (QUBO) problem on a graph $G$ can be solved 
  using an adiabatic quantum computer that implements an Ising spin-1/2
  Hamiltonian, by reduction through {\em minor-embedding} of
  $G$ in the quantum hardware graph $\uni$.
  There are two components to this reduction: {\em embedding} and {\em
    parameter setting}. The embedding problem is to find a minor-embedding
  $\Gemb$ 
  of a graph $G$ in $\uni$, which is a
  subgraph of $\uni$ such that $G$ can be obtained from $\Gemb$ by
  contracting edges. 
 The parameter setting problem is to determine the
  corresponding parameters, qubit biases and coupler strengths, of the embedded
  Ising Hamiltonian.
  In this paper, we focus on the parameter setting problem.
  As an example, we demonstrate the embedded Ising Hamiltonian for solving
  the maximum independent 
  set (MIS) problem via  adiabatic quantum computation (AQC) using an Ising
  spin-1/2 system.
We close by discussing several related algorithmic problems that need to be
investigated in order to facilitate the design of adiabatic algorithms and
AQC architectures.
\end{abstract}

\section{Introduction}
\label{sec:introduction}
Adiabatic Quantum Computation (AQC) was proposed by Farhi
et~al.\cite{FGGS00,FGGLLP01} in 2000. 
The AQC model is based on the {\em adiabatic
theorem} (see, e.g. \cite{Reichardt-04}).  Informally, the theorem says that if we take a
quantum system whose Hamiltonian ``slowly'' changes from $\ham_{\ms{init}}$ (initial
Hamiltonian) to $\ham_{\ms{final}}$ (final Hamiltonian), 
then if we start with the system in the {\em groundstate}
(eigenvector corresponding to the lowest
eigenvalue) of $\ham_{\ms{init}}$, then at the end of the evolution the
system will be ``predominantly'' in the ground
state of $\ham_{\ms{final}}$. The theorem is
used to construct {\em adiabatic algorithms} for optimization problems in
the following way:  The initial Hamiltonian $\ham_{\ms{init}}$ is designed such
that the system can be readily initialized into its known groundstate, while the
groundstate of the final Hamiltonian $\ham_{\ms{final}}$ encodes the
answer to the desired optimization problem.
 The complete (or {\em system}) Hamiltonian at a time
$t$ is then given by
\begin{equation*}
\ham(t) = \biggl(1-s(\frac{t}{T})\biggr)\ham_{\ms{init}} + s(\frac{t}{T}) \ham_{\ms{final}}
\end{equation*}
 for $t \in [0,T]$ where $s$ increases monotonically from $s(0)=0$ to $s(1)=1$ and
$T$ is the total running time of the algorithm. If $T$ is large
 enough, which is determined by the minimum spectral gap (the difference between the two lowest energy levels) of
the system Hamiltonian, the adiabatic theorem guarantees the  state at time $t$
 will be the groundstate of $\ham(t)$, leading to the solution, the ground
 state of $\ham(T)=\ham_{\ms{final}}$. 

It is believed that AQC is advantageous over standard (the gate model) 
quantum computation in that it is more robust against environmental noise~\cite{CFP02,amin-2008-100,amin-2008}.
In 2004, D-Wave Systems Inc. undertook the endeavor to build an 
adiabatic quantum computer for solving NP-hard problems. 
Kaminsky and Lloyd~\cite{KL02} 
proposed a scalable architecture for AQC of NP-hard problems. 
Note that AQC is polynomially equivalent to the standard quantum
computation, and therefore AQC is universal. 
Farhi
et~al. in their original paper~\cite{FGGS00} showed that AQC can be
efficiently simulated by standard quantum computers.  
In 2004, Aharonov~et~al~\cite{ADKLLR04} proved the more difficult converse
statement that standard quantum computation can be efficiently simulated by
AQC. 
 In this
paper, however, we will focus on a subclass of Hamiltonians, known as 
{\em Ising models in a transverse field}, that are NP-hard but not
universal for quantum computation. This subclass of Hamiltonians has been 
implemented by D-Wave Systems Inc.

D-Wave's quantum hardware architecture can be viewed as an undirected
graph $\uni$ with weighted vertices and weighted edges. See
Figure~\ref{fig:extendedGrid} for an example.
\begin{figure}[h]
\centering{
\includegraphics{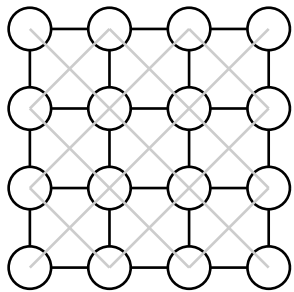}
} 
\caption{An example hardware graph: a $4 \times 4$ extended grid. Each qubit
  is coupled with its nearest and next-nearest neighbors.}
\label{fig:extendedGrid}
\end{figure}
Denote the vertex set of $\uni$ by
$\ver(\uni)$ and the edge set of $\uni$ by $\edge(\uni)$.
Each vertex $i \in \ver(\uni)$ corresponds to a qubit, and each edge $ij
\in \edge(\uni)$ corresponds to a coupler between qubit $i$ and qubit
$j$.
In the following, we will use qubit and vertex, and coupler and edge
interchangeably when there is no confusion.
There are two weights, $h_i$ (called the qubit {\em bias}) and $\Delta_i$
(called the {\em tunneling amplitude}), associated
with each qubit $i$. There is a weight $J_{ij}$ (called the coupler {\em
  strength})
 associated with each coupler
$ij$. In general, these weights are functions of time, i.e., they vary
over the time, e.g., $h_i(t)$. 

The Hamiltonian of the Ising model in a
  transverse field thus implemented is:
\begin{equation*}
\ham(t) = \sum_{i \in \ver(G)} h_i(t) \sigma^z_i + \sum_{ij \in \edge(G)} J_{ij}(t)
\sigma^z_i \sigma^z_j + \sum_{i \in \ver(G)} \Delta_i(t) \sigma^x_i 
\label{IsingEq}
\end{equation*}
with $\sigma_i^z = I \otimes I \otimes \ldots \otimes \sigma^z \otimes
\ldots \otimes I$ (the $\sigma^z$ is in the $i$th position), similarly for
$\sigma^x_i $ and $\sigma^z_i\sigma^z_j$,
where
$I$  is the $2 \times 2$ identity
matrix, while
$\sigma^z$ and $\sigma^x$ are the Pauli matrices given by
\begin{equation*}
\sigma^z = \begin{bmatrix} 1 & 0 \\ 0 & -1\end{bmatrix}, \quad \text{and} \quad \sigma^x = \begin{bmatrix} 0 & 1 \\ 1 & 0 \end{bmatrix}.
\end{equation*}

In general, the transverse field (i.e. $ \sum_{i \in \ver(G)} \Delta_i
\sigma^x_i$) encodes the initial Hamiltonian $\ham_{\ms{init}}$, while the
other two terms encode the final Hamiltonian:
\begin{equation}
\ham_{\ms{final}} = \sum_{i \in \ver(G)} h_i \sigma^z_i + \sum_{ij \in \edge(G)} J_{ij}
\sigma^z_i \sigma^z_j.
\label{eqn:final-ham}
\end{equation}
One can show (see Appendix) that the eigenvalues and the corresponding eigenstates of
$\ham_{\ms{final}}$ are
encoded in the following energy function of Ising Model:
\begin{eqnarray}
\energy(s_1,\ldots, s_n) = \sum_{i \in \ver(G)} h_i s_i + \sum_{ij \in\edge(G)} J_{ij}
      s_is_j
\label{eqn:OE}
\end{eqnarray}
where $s_i \in \{-1,+1\}$, called a {\em spin}, and  $h_i, J_{ij} \in \reals$.
In particular, the smallest eigenvalue of $\ham_{\ms{final}}$ corresponding to the minimum
 of $\energy$, and $\argmin \energy$ corresponds to its eigenvector (called {\em ground
  state}) of $\ham_{\ms{final}}$. 
When there is no confusion, we use the energy function of Ising model
and Ising Hamiltonian interchangeably. Hereafter, we refer the problem of
finding the minimum 
energy of the Ising Model or
equivalently the groundstate of Ising Hamiltonian as the {\em Ising problem}.
It can be shown that the Ising problem is equivalent to the problem of
Quadratic
Unconstrained Boolean
Optimization (QUBO)(see~\cite{BHT06} and references therein), which has
been shown to be a common model for a wide variety of discrete
optimization problems. More specifically, finding the minimum of $\energy$
in \eqref{eqn:OE} is equivalent to finding the maximum of the following
function 
(which is also known as
    quadratic pseudo-Boolean function~\cite{BH02})
of QUBO on the same graph: 
\begin{eqnarray}
\oy(x_1, \ldots, x_n) &=& \sum_{i \in \ver(G)} c_i x_i - \sum_{ij \in\edge(G)} J_{ij}
  x_ix_j 
\label{eqn:OY}
\end{eqnarray}
where $x_i \in  \{0,1\}$, $c_i, J_{ij} \in \reals$. 
The
correspondence between the parameters, $h$s and $c$s
will be shown in Section~\ref{sec:QUBO-Ising}. 
Therefore, given an Ising/QUBO problem on graph $G$, one can thus solve the
problem on an  adiabatic quantum computer (using an Ising spin-1/2 system)
if $G$ can be embedded as a subgraph of the quantum hardware graph $\uni$. We refer this
embedding problem as {\em subgraph-embedding},
to be
defined formally in Section~\ref{sec:embed}.
In general, there are physical constraints on the hardware
graph $\uni$. In particular, there is a {\em degree-constraint} in that
each qubit can have at most a constant number 
of couplers dictated by hardware design. Therefore, besides the possible difficulty of the
subgraph-embedding problem\footnote{Readers should be cautious not to confuse
  this embedding problem with the NP-complete subgraph isomorphism problem,
  in which both graphs are unknown. However, in our case, the hardware
  graph is known. For example, if the hardware graph is a complete graph,
  then the embedding problem will be trivial.}, the graphs that can be
solved on a given hardware graph $\uni$ through
subgraph-embedding must also be degree-bounded. 
Kaminsky~et~al.~\cite{KL02,KLO04}
observed and proposed that one can embed $G$ in $U$ through
ferromagnetic coupling ``dummy vertices'' to solve 
Maximum Independent Set (MIS) problem \footnote{MIS is a special case of QUBO and
  will be addressed in Section~\ref{sec:wmis}.} of planar cubic graphs (regular graphs of
degree-3)\footnote{In their earlier paper~\cite{KL02}, it was said for
  graphs with degree at most 3, but the Ising Hamiltonian they used there was for
  regular graphs of degree-3.} on an adiabatic quantum computer. 
In particular, they proposed an $n \times n$
square lattice as a scalable hardware architecture on which all $n/3$-vertex planar cubic graphs are
embeddable.  
The notion of embedding here follows naturally from physicists' intuition
that 
each {\em logical qubit} (corresponding to a vertex in the input
graph) is mapped to a subtree of {\em physical qubits} (corresponding to vertices in
the hardware graph) that are ferromagnetically coupled such 
that each
subtree of physical qubits acts like a single logical qubit.
For example, in Figure~\ref{fig:example-embedding}, the logical qubit $1$
(in orange color) of the graph $G$ is mapped to a subtree of physical
qubits (labelled $1$) of the square lattice. 
\begin{figure}[h]
\centering{
\includegraphics[angle=90]{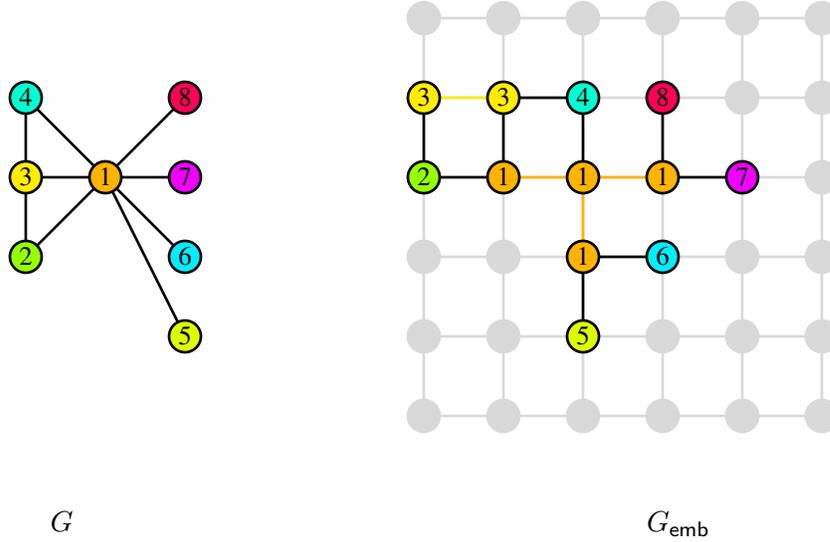}
$G$\hspace*{3in}$\Gemb$
} 
\caption{$\Gemb$(right) is a minor-embedding of $G$(left) in the square
  lattice $\uni$. Each vertex (called a logical qubit) of $G$ is mapped to a
  (connected) subtree of (same color/label) vertices (called physical qubits) of $\uni$. $G$
  is called a (graph) minor of $\uni$.}
\label{fig:example-embedding}
\end{figure}
Informally, a minor-embedding $\Gemb$ of a graph $G$ in the hardware graph
$\uni$ is a subgraph of $\uni$ such that $\Gemb$ is an ``expansion'' of $G$
by replacing each vertex of $G$ with a (connected) subtree of $\uni$, or
equivalently, $G$ can be obtained from $\Gemb$ by contracting edges (same color
in Figure~\ref{fig:example-embedding}). In
graph theory, $G$ is called a (graph) {\em minor} of $\uni$ (see for
example~\cite{Distel05}). The minor-embedding will be formally defined 
in Section~\ref{sec:embed}.
(Remark: The embedding in \cite{KL02,KLO04} is a special case of minor-embedding,
known as {\em topological-minor} embedding.)

By {\em reduction
through minor-embedding}, we mean that one can reduce the original Ising Hamiltonian on the input
graph $G$ to the embedded Ising Hamiltonian
$\ham^{\emb}$ on its minor-embedding $\Gemb$, i.e., the solution to the
embedded Ising Hamiltonian gives rise to the solution to the original
Ising Hamiltonian.
The intuition suggests that the reduction will be correct provided that the
ferromagnetic coupler strengths used are sufficiently strong (i.e., large negative number).
However, how strong is ``strong enough''? 
In~\cite{KL02,KLO04}, they do not address this question, i.e. what are the
required strengths 
of these ferromagnetic couplers?
In Section~\ref{sec:easy-upper}, we will show that it is not difficult to
give an upper bound for the ferromagnetic coupler strengths and thus explain the 
intuition. However,
there are indications~\cite{AC08} that too
strong ferromagnetic coupler strengths might slow down the adiabatic
algorithm. 
Furthermore, an adiabatic quantum computer is an analog computer and analog
parameters can only be set to a certain degree of precision (a condition
much more stringent
than the setting of digital parameters). 
Hence, the allowed values
of coupler strengths are limited. 
Therefore, from the
computational point of view, it is important 
to derive as small (in terms of magnitude) as possible sufficient condition
for these ferromagnetic coupler strengths.
Furthermore, what should  the
bias for physical qubits be? 

There are two components to
the reduction: {\em embedding} and {\em parameter setting}.
The embedding problem is to find a minor-embedding
  $\Gemb$ 
  of a graph $G$ in $\uni$. 
This problem is interdependent of the hardware graph design problem, which will be discussed in
Section~\ref{sec:discussion}.
The parameter setting problem is to set the
  corresponding parameters, qubit bias and coupler strengths, of the embedded
  Ising Hamiltonian.
In this paper, we assume that the minor-embedding $\Gemb$ is given, and
focus on the parameter setting problem (of the final Hamiltonian).
Note that there are two aspects of efficiency of a reduction. One is how
efficient one can reduce the original problem to the reduced
problem. For example, here we are concerned how efficiently we can compute
the minor-embedding and how efficiently we can compute the new parameters of
the embedded Ising Hamiltonian.
The other aspect concerns about the
  efficiency (in terms the running time) of the adiabatic algorithm for the
  reduced problem. 
In general, the latter depends on the former. For example, the running time
of the adiabatic algorithm may depend on a ``good'' embedding that is
reduced to. 
According to the adiabatic theorem (see, e.g. \cite{Reichardt-04}), 
the running time of the adiabatic algorithm depends on
the minimum spectral gap (the difference between the two lowest energy levels) of
the system Hamiltonian, which is defined by both initial Hamiltonian and final
Hamiltonian. 
That is, the running time (and thus the efficiency of the reduction) will
depend on both the initial Hamiltonian and the final Hamiltonian. 
In this paper, our focus is only on the final Hamiltonian, and therefore we
are not able to address the running time of the adiabatic
algorithm. Furthermore, the estimation of the minimum spectral gap of the
(system) Hamiltonian is in general hard. Consequently, analytically
analyzing the
running time of an adiabatic algorithm is in general an open question. 

Finally, let us remark that there is another different approach based on  perturbation theory
 by Oliveira \& Terhal~\cite{oliveira-2005} for performing the reduction.
 In particular, 
 they employed perturbative gadgets  to reduce a 2-local
(system) Hamiltonian to a 2-local Hamiltonian on a 2-D square lattice, and
  were able to show (as in the pioneering work \cite{kempe-2006-35}) that the minimum spectral gap
 (and thus the running time) of the system Hamiltonian is preserved (up to
 a polynomial factor) after the reduction, for any given initial
 Hamiltonian.
However, besides the
ineffective embedding, as pointed out in \cite{bravyi-2008}, the method is
``unphysical'' 
as it requires that each parameter
grows with the system size. 
\vnote{Is it possible that by choosing the embedding-dependent initial
  Hamiltonian, one can get a better running time?}

The rest of the paper is organized as follows. 
In
Section~\ref{sec:QUBO-Ising}, we recall the equivalences between the QUBO problem
and the Ising problem. In Section~\ref{sec:embed}, we introduce the
minor-embedding definition and mention  related work in graph theory. 
In Section~\ref{sec:parameters}, we derive the new parameters, namely
values for the qubit bias and the
sufficient condition for the ferromagnetic coupler strengths, for the embedded Ising Hamiltonian such that
the original Ising problem can be correctly solved through the embedded
Ising problem. 
In Section~\ref{sec:wmis}, we show the embedded Ising Hamiltonian
for solving the (weighted) MIS problem. 
Finally, we conclude with several related algorithmic problems  that need
to be investigated in order to facilitate the design of efficient adiabatic algorithms and
AQC architectures in Section~\ref{sec:discussion}.

\section{Equivalences Between QUBO
  and the Ising Problem}
\label{sec:QUBO-Ising}
In this section, we recall the
equivalences between the problem of QUBO (maximization of $\oy$ in Eq.~\eqref{eqn:OY})
and the Ising problem (minimization of $\energy$ in Eq.~\eqref{eqn:OE}).
Notice that $x_i = \frac{s_i+1}{2}$, that is, $x_i=1$ corresponds to
$s_i=1$, and $x_i=0$ corresponds to $s_i=-1$.
Using a change of variables, we have
\begin{eqnarray*}
\oy(x_1, \ldots, x_n) 
  &=& \frac{1}{2} \sum_{i \in \ver(G)} c_i - \frac{1}{4} \sum_{ij \in\edge(G)}J_{ij}  - \frac{1}{4} \left(\sum_{i \in \ver(G)}
  \left(\sum_{j \in \nbr(i)} J_{ij} - 2c_i \right)s_i + \sum_{ij \in \edge(G)} J_{ij} s_i s_j\right)
\end{eqnarray*}
where $\nbr(i) \mdef \{j: ij \in \edge(G)\}$, the neighborhood of vertex
$i$, for $i \in \ver(G)$.

Therefore, $\ms{Max} \oy(x_1, \ldots, x_n)$ in Eq.~\eqref{eqn:OY} 
is equivalent to $\ms{Min} \energy(s_1,\ldots, s_n)$ in Eq.~\eqref{eqn:OE}
where $h_i =\sum_{j \in \nbr(i)} J_{ij} - 2 c_i$,  or 
$c_i=1/2(\sum_{j \in
  \nbr(i)} J_{ij} - h_i)$
for $i \in \ver(G)$.
See Figure~\ref{fig:new-parameters} for the correspondences between
parameters in QUBO and the Ising model.
\begin{figure}[h]
\centering{
\includegraphics[width=5in]{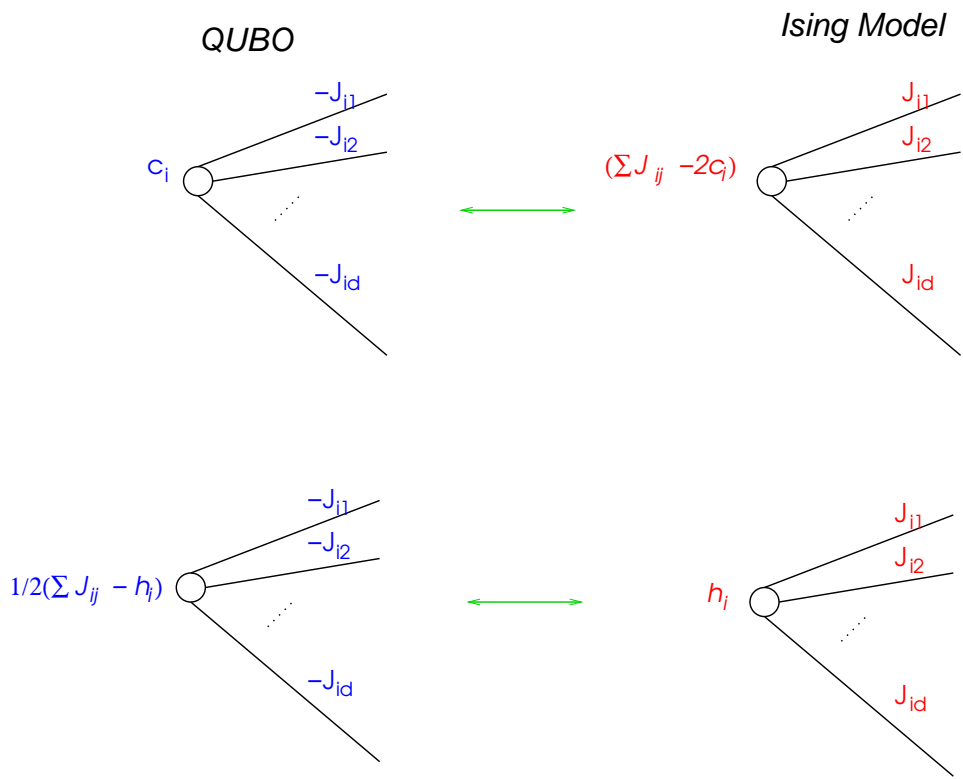}
} 
\caption{The correspondences between the parameters in QUBO and the Ising Model.}
\label{fig:new-parameters}
\end{figure}

\section{Minor-Embedding}
\label{sec:embed}
\begin{definition}
Let $\uni$ be a fixed hardware graph. Given $G$, 
the {\em minor-embedding} of $G$ is
defined by 
$$\phi: G \longrightarrow \uni$$ such that
\begin{itemize}
\item each vertex in $\ver(G)$ is mapped to a connected subtree $T_i$ of $\uni$;
\item there exists a map $\tau:
  \ver(G) \times \ver(G) \longrightarrow \ver(\uni)$
such that 
for each $ij \in \edge(G)$, there are corresponding $i_{\tau(i,j)}
  \in \ver(T_i)$ and $j_{\tau(j,i)} \in \ver(T_j)$
  with  $i_{\tau(i,j)}j_{\tau(j,i)} \in \edge(\uni)$.
\end{itemize}  
Given $G$, if $\phi$ exists, we say that $G$ is {\em
  embeddable}
 in $\uni$. In graph theory, $G$ is called a {\em minor} of $\uni$. 
When $\phi$ is clear from the context, we denote the minor-embedding
  $\phi(G)$ of $G$
  by $\Gemb$.
\end{definition} 

Equivalently, one can think of a minor $G$ of $\uni$ as a graph that can be obtained from a subgraph of
  $\uni$ by contracting edges.
See Figure~\ref{fig:example-embedding} for an example.

In particular, there are two special cases of minor-embedding:
\begin{itemize}
\item{\em Subgraph-embedding:} Each $T_i$ consists of a single vertex in
  $\uni$. That is, $G$ is isomorphic to $\Gemb$ (a subgraph of $\uni$).
\item{\em Topological-minor-embedding:} Each $T_i$ is a chain (or path) of vertices in $\uni$.
\end{itemize}

Minors are well-studied in graph theory, see for
example~\cite{Distel05}. 
Given a {\em fixed} graph $G$, there are algorithms that find a minor-embedding of $G$ in
$\uni$ in polynomial time of size of $\uni$, from the pioneering
$O(|\ver(\uni)|^3)$ time algorithm by Robertson and Seymour~\cite{RS95} to recent nearly linear
time algorithm of B.~Reed (not yet published). However, it is worthwhile to reiterate that these
algorithms are for {\em fixed} $G$, and their running times are {\em exponential} in the
size of $G$. 
Here the minor-embedding problem is to find a minor-embedding of $G$ (for any
given $G$) while
fixing $\uni$. 
To the best of our knowledge, the
only known work related to our minor-embedding problem was by Kleinberg and Rubinfeld~\cite{KR96}, in
which they showed that there is a randomized polynomial algorithm, based on a
random walk, to find a minor-embedding in a given degree-bounded expander. 

Our embedding problem might appear similar to the embedding
  problem from parallel architecture studies. However, besides the
  different physical constraints for the design of architectures, the requirements are very
  different. In particular, in our embedding problem, we do not allow for
  {\em load} $>1$, which is the maximum number of logical qubits mapped to
  a single physical qubit. 
  Also, we require {\em dilation}, which is the
  maximum number of stretched edges (through other qubits), to be exactly 1.
However, all of the existing research on embedding problems for parallel
  processors~\cite{Leighton-book},
at least one of the conditions  is violated (namely either load
  $>1$ or dilation $>1$).
In
this paper, our focus is on 
parameter setting such that the reduction is correct, and not on the 
minor-embedding algorithm and/or the related (minor-universal) 
hardware graph
$\uni$ design. These problems
will be addressed in a subsequent paper.

\section{Parameter Requirement for the Embedded Ising Hamiltonian}
\label{sec:parameters}
Let the QUBO problem, specified by $\oy$ and $G$,  be given as in
Eq.\eqref{eqn:OY},
 and the corresponding $\energy$ in
 Eq.\eqref{eqn:OE}. 
Suppose $\Gemb$ is a
minor-embedding of $G$, and let $\energy^{\emb}$ be the embedded energy
function  
associated with $\Gemb$:
\begin{eqnarray*}
  \energy^{\emb} (s_1, \ldots, s_N) &=& \sum_{i \in \ver(\Gemb)} h'_i s_i + \sum_{ij \in
      \edge(\Gemb)} J'_{ij} s_i s_j
\end{eqnarray*}
where $|\ver(\Gemb)| =N$.

Our goal is to find out the requirement for 
new parameters, $h'$s and $J'$s, such that one can solve the original Ising
problem on $G$ by solving the embedded Ising problem on its embedding
$\Gemb$, or equivalently, 
such that there is an one-one
correspondence between the minimum of $\energy$ (and thus the maximum of
$\oy$) 
and the minimum of $\energy^{\emb}$.

As we mentioned in the introduction, the idea that we can solve the original
Ising problem (i.e. finding the groundstate of the Ising Hamiltonian) on the input graph $G$ by solving the new Ising
problem on the embedded graph $\Gemb$ in $\uni$ is that one can use ferromagnetic
couplers to connect the physical qubits in each $T_i$ of $\Gemb$ such that
the subtree $T_i$ will act as one logical qubit $i$ of $G$. 


\paragraph{Notation.} First, we recall that 
$\ver(\Gemb) = \bigcup_{i \in \ver(G)} \ver(T_i)$ and
$\edge(\Gemb) =  \bigcup_{i \in \ver(G)} \edge(T_i) \union \bigcup_{ij
  \in \edge(G)}i_{\tau(i,j)}j_{\tau(j,i)}.$
We distinguish the edges within $T_i$s
from the edges corresponding to the edges in the original graph. 
Denote the latter by $\Oedge(\Gemb)$, that is, 
$\Oedge(\Gemb) = \bigcup_{ij \in \edge(G)} i_{\tau(i,j)}j_{\tau(j,i)}$.
(The black edges in $\Gemb$ of Figure~\ref{fig:example-embedding} correspond to $\Oedge(\Gemb)$.)
For convenience, $\onbr(i_k) \mdef \nbr(i_k) \cap \Oedge(\Gemb)$
and $\fnbr(i_k) \mdef \nbr(i_k) \setminus \onbr(i_k)$,
$i_k \in \ver(T_i)$. Note $j_{\tau(j,i)} \in \onbr(i_k) \Leftrightarrow ij
\in \edge(G)$.

The above intuition 
suggests that  we use  
the same coupler strength for each original edge, i.e., 
$J'_{i_{\tau(i,j)}j_{\tau(j,i)}} = J_{ij}$  for ${i_{\tau(i,j)}j_{\tau(j,i)}} \in
\Oedge(\Gemb)$, and 
use ferromagnetic coupler strength, $F_i^e(<0)$, for each edge
$e \in \edge(T_i)$, and redistribute the bias $h_i$ of a logical qubit $i$ to 
its physical qubits in $T_i$. That is, we choose $h'_{i_k}$  for physical
qubit $i_k \in \ver(T_i)$ such that 
$\sum_{{i_k} \in \ver(T_i)} h'_{i_k} = h_i$.

Therefore, we have 
\begin{eqnarray}
    \energy^{\emb}(s_1, \ldots, s_N)
    &=& \sum_{i \in \ver(G)} \left( \sum_{i_k \in \ver(T_i)} h'_{i_k} s_{i_k} + \sum_{i_pi_q
      \in \edge(T_i)} F_i^{pq}s_{i_p}s_{i_q} \right) + \sum_{ij \in \edge(G)} J_{ij} 
s_{i_{\tau(i,j)}} s_{j_{\tau(j,i)}}
\label{eqn:Eemb-general}
\end{eqnarray}

\subsection{An Easy Upper Bound for the Ferromagnetic Coupler Strengths}
\label{sec:easy-upper}
In this section, we derive an easy upper bound for the ferromagnetic coupler
strengths.
The derivation is based on the {\em penalty or multiplier} method, in
which a constrained optimization is reduced to an unconstrained one
by replacing the
corresponding expression into the objective function as a penalty term with
a large multiplier.\footnote{This is closely related to the discrete Lagragian Multiplier method. However,
the Lagragian multipliers are treated as unknown or iteratively solved
with the Lagragian function.}
For example, this was used in the reduction
from general (higher-order) unconstrained binary optimization to the
quadratic one (see \cite{BH02}).

Notice that by construction (i.e., $G$ is a minor of $\Gemb$ and $h'_{i_k}$s are
chosen such that $\sum_{{i_k} \in \ver(T_i)} h'_{i_k} =
h_i$), minimization of $\energy$ in \eqref{eqn:OE} is equivalent to 
\begin{eqnarray}
\label{eqn:constrain-E}
\min \energy_{\ms{constrained}}(s_1, \ldots, s_N) &=&
\sum_{i \in \ver(G)} \sum_{i_k \in \ver(T_i)} h'_{i_k} s_{i_k} + \sum_{ij \in \edge(G)} J_{ij} 
s_{i_{\tau(i,j)}} s_{j_{\tau(j,i)}} \\
\mbox{ subject to } && \nonumber\\
 s_{i_p}s_{i_q} &=& 1 \mbox{ for all } i_pi_q \in \edge(T_i), i
\in \ver(G)\nonumber
\end{eqnarray}
where 
the condition $s_{i_p}s_{i_q} = 1$ for
all $i_pi_q \in \edge(T_i)$
is equivalent to 
requiring the spins  of physical qubits that correspond to the
same logical qubit $i$
to be of the same sign.


The corresponding unconstrained minimization is thus
\begin{eqnarray}
\label{eqn:unconstrain-E}
\min  \energy'(s_1, \ldots, s_N) &=&
\sum_{i \in \ver(G)} \sum_{i_k \in \ver(T_i)} \left( h'_{i_k} s_{i_k} + \sum_{i_pi_q
      \in \edge(T_i)} F_i^{pq}(s_{i_p}s_{i_q}-1) \right) + \sum_{ij \in \edge(G)} J_{ij} 
s_{i_{\tau(i,j)}} s_{j_{\tau(j,i)}}
\end{eqnarray}
which is equivalent to solving Eq.\eqref{eqn:Eemb-general} as $\sum_{i \in \ver(G)} \sum_{i_pi_q
      \in \edge(T_i)} F_i^{pq}$ is a constant. 
We are interested in how large (in terms of magnitude) $F_i^{pq}$s are
sufficient to guarantee  that 
the solution of Eq.\eqref{eqn:unconstrain-E} gives the solution to  Eq.\eqref{eqn:constrain-E}, 
 and consequently to Eq.\eqref{eqn:OE}.
The result is stated in the following theorem.
\begin{theorem}
Let   $\energy_{\ms{constrained}} $  and $\energy'$  given as in
Eq.\eqref{eqn:constrain-E} and Eq.\eqref{eqn:unconstrain-E}. 
Suppose that for $i \in \ver(G)$,  
\begin{equation}
  F_i^{pq} < - \left(|h_i| + \sum_{j \in
  \nbr(i)}|J_{ij}| \right)~\quad\quad \text{ for } i_pi_q \in \edge(T_i)
\label{eq:F-condition}
\end{equation}
then we have $s_{i_p}^*s_{i_q}^* = 1 \mbox{ for all } i_pi_q \in \edge(T_i)$, $i \in
\ver(G)$, where $(s_{1_1}^*, \ldots, s_{N}^*) = \argmin \energy'$.
Consequently, 
$$\min \energy_{\ms{constrained}}(s_1, \ldots, s_N) = \min \energy'(s_1,
\ldots, s_N).$$
\label{thm:easy}
\end{theorem}

\begin{proof}
Let $(s_{1_1}^*, \ldots, s_{N}^*) = \argmin \energy'.$
Suppose on the contrary. That is, there exists $f_pf_q \in \edge(T_f)$, for
some $1 \le f \le n$ such that $s_{f_p}^*s_{f_q}^* =-1$. Then we have
$$ \energy'(s_1^*, \ldots, s_N^*) \ge - \left(\sum_{f_k \in \ver(T_f)}
  |h'_{f_k}| + \sum_{j \in \nbr(f)}|J_{fj}|\right) - 2 F_f^{pq} +
  \energy'_{\ms{rest}}(s_1^*, \ldots, s_N^*)$$ 
where $$\energy'_{\ms{rest}}(s_1, \ldots, s_N) = \sum_{i \in \ver(G), i\neq f} \sum_{i_k \in \ver(T_i)} \left( h'_{i_k} s_{i_k} + \sum_{i_pi_q
      \in \edge(T_i)} F_i^{pq}(s_{i_p}s_{i_q}-1) \right) + \sum_{ij \in
  \edge(G), i\neq f} J_{ij} 
s_{i_{\tau(i,j)}} s_{j_{\tau(j,i)}}.
$$
Thus, if $F$s satisfy Eq.\eqref{eq:F-condition2}, we have $$\energy'(s_1^*, \ldots, s_N^*) > \sum_{f_k \in \ver(T_f)}
  |h'_{f_k}| + \sum_{j \in \nbr(f)}|J_{fj}|  +  \energy'_{\ms{rest}}(s_1^*,
  \ldots, s_N^*)$$
which is at least $\min  \energy'(s_1, \ldots, s_N) = \energy'(s_1^*,
\ldots, s_N^*)$, a contradiction. 
Consequently,  we can conclude that $s_{i_p}^*s_{i_q}^* = 1 \mbox{ for all }
i_pi_q \in \edge(T_i)$, $i \in
\ver(G)$, where $(s_{1_1}^*, \ldots, s_{N}^*) = \argmin \energy'$,
implying the claimed equality of the minima.
\end{proof}

\subsection{A Tighter Bound for the Ferromagnetic Coupler Strengths}
 As we mentioned Section~\ref{sec:introduction}, 
 it is desired to obtain tighter bounds for the ferromagnetic coupler strengths.
 In this section, we show that with a more careful analysis, we can reduce
 the bound by setting  the qubit bias ($h'$)
appropriately.

For $i \in \ver(G)$, 
 let $$C_i \mdef \sum_{j \in \nbr(i)} |J_{ij}| -
|h_i|.$$ 
Observe that if $C_i<0$, that is, $\sum_{j \in \nbr(i)} |J_{ij}| <
|h_i|$. Then we have $s_i^*=-1$ for $h_i>0$, and $s_i^*=+1$ for $h_i<0$,
where $(s_1^*, \ldots, s_n^*) = \argmin \energy$. Therefore, WLOG, for the
rest of the paper, we will assume that $C_i \ge 0$.

\begin{theorem}
Let $\energy$  ($\oy$ resp.) and  $G$ be given as in Eq.\eqref{eqn:OE}
(Eq.\eqref{eqn:OY} resp.). Suppose $\Gemb$ is a
minor-embedding of $G$, 
and let $\energy^{\emb}$ be the energy of the embedded Ising Hamiltonian
given in  Eq.\eqref{eqn:Eemb-general}.
Then
for all $i \in \ver(G)$,
if 
$$ h'_{i_k} = sign(h_i) \left\{
\begin{array}{ll}
  \sum_{j_{\tau(j,i)} \in \onbr(i_k)} |J_{ij}| - C_i/l_i & i_k \mbox{ is one of the
    $l_i$ leaves of } T_i\\
  \sum_{j_{\tau(j,i)} \in \onbr(i_k)} |J_{ij}|  & \mbox{ otherwise }
\end{array}
\right.
$$ 
and
$$F_i^{e} < - \frac{l_i-1}{l_i} C_i ~~ \mbox{for all } e \in \edge(T_i), $$
where $C_i \ge 0$ (defined above)
and $sign(h_i) = +1 (-1$ resp.) if $h_i \ge 0$ ($h_i<0$ resp.),
we have $s_{i_p}^*s_{i_q}^* = 1 \mbox{ for all } i_pi_q \in \edge(T_i)$, $i \in
\ver(G)$, where $(s_{1_1}^*, \ldots, s_{N}^*) = \argmin \energy^{\emb}$.
Consequently, 
there is one-one correspondence between  $\argmin
\energy^{\emb}$ and $\argmin \energy$ (and thus $\argmax \oy$).

Furthermore, if we set $F_i^{e} = - \frac{l_i-1}{l_i} C_i - g_i/2$, for some
$g_i>0$, then the
spectral gap (which is the difference between the two lowest energy levels) of the embedded Ising Hamiltonian will be the minimum
of $\min_{i \in \ver(G)} g_i$ and the spectral gap of the original Ising
Hamiltonian. 
\label{thm:main}
\end{theorem}
\vnote{More elaboration about $\Delta_i$.}
See Figure~\ref{fig:new-parameters} for an example of the  corresponding
parameters in the embedded Ising Hamiltonian (for the topological-minor-embedding case). 
\begin{figure}[h]
\centering{
\includegraphics[width=5in]{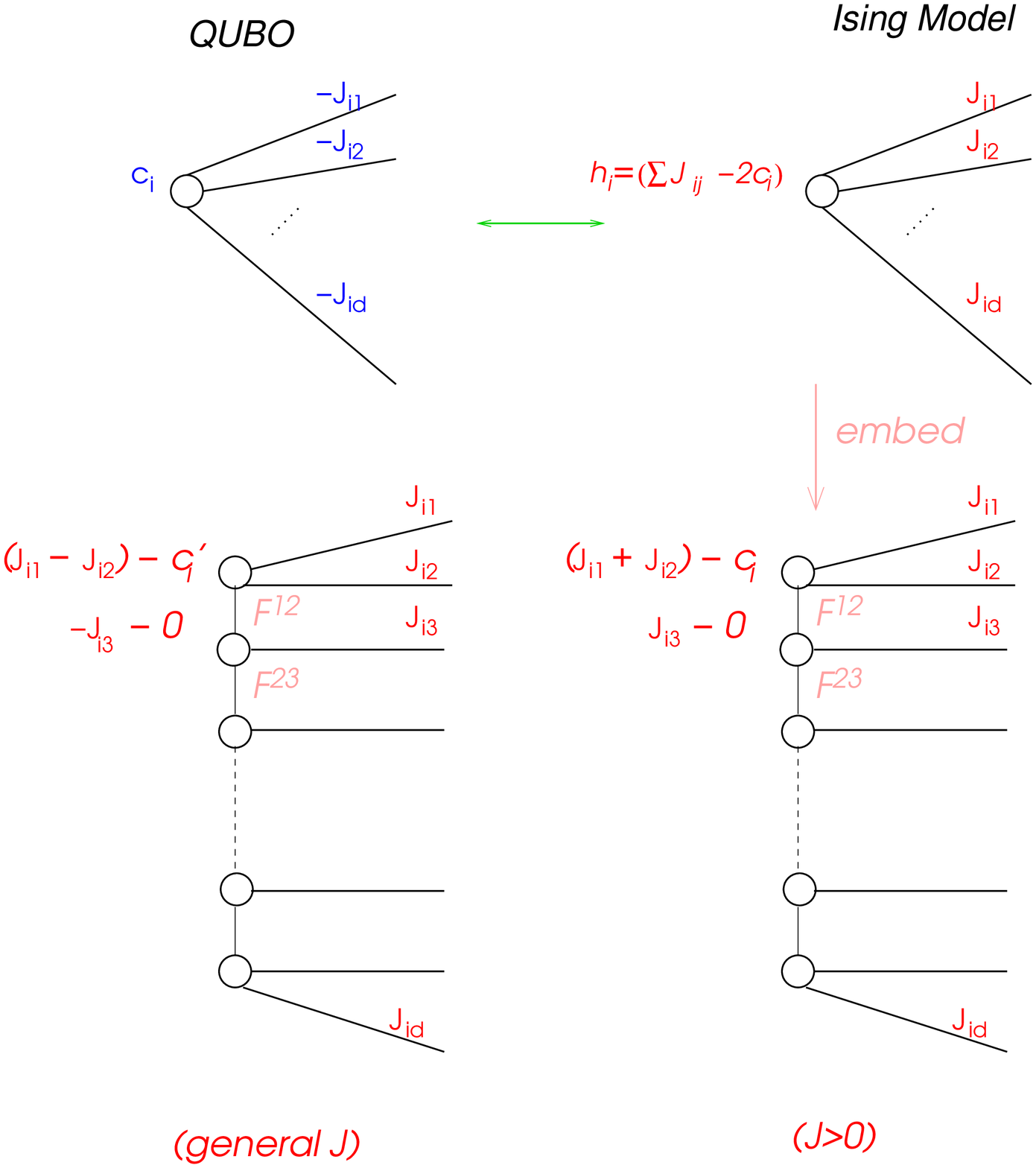}
} 
\caption{The corresponding parameters in the embedded Ising Model. Bottom
  right, for $J>0$; bottom left, for general $J$: $J_{i1}>0$, $J_{i2}<0,
  J_{i3}<0$, and  $c_i'= c_i - \sum_{j \in \nbr(i), J_{ij}<0} J_{ij}$.}
\label{fig:new-parameters}
\end{figure}

In the following, we first explain the main idea behind, followed
by the formal proof.
 For the
illustration purpose, suppose $h_i>0$, consider the simplest case in which all but one leaf is $-1$.
Now, consider the energy change if
the leaf is flipping from $+1$ to $-1$. Our goal is to set $F$s
as small (in terms of the magnitude) as possible, such that the energy
change is at least greater than zero. Note that the energy change
$\Delta\energy_{i_k} = $
$ 2 (h'_{i_k} + \sum_{j\tau(i,j) \in
  \onbr(i_k)}\delta_{j_{\tau(j,i)}}J_{ij} - F)$ where $\delta_{j_{\tau(j,i)}}$
depends on the 
sign of spin $s_{j_{\tau(j,i)}}^*$. We would like to bound
$\Delta\energy_{i_k}$ without knowing the signs of
$s_{j_{\tau(j,i)}}^*$s.
Observe that if $J$s are all positive,
then the worst case we have 
$\Delta\energy_{i_k} \ge 2 (h'_{i_k} - \sum_{j_{\tau(i,j)} \in
  \onbr(i_k)} J_{ij} -F)$. And in general 
$\Delta\energy_{i_k} \ge 2 (h'_{i_k} - \sum_{j_{\tau(i,j)} \in
  \onbr(i_k)} J_{ij}^+ + \sum_{j_{\tau(i,j)} \in
  \onbr(i_k)} J_{ij}^-  -F)$, where $J_{ij}^+>0$ and $J_{ij}^-<0$. Consequently, setting $F<h'_{i_k} - \sum_{j_{\tau(i,j)} \in
  \onbr(i_k)} |J_{ij}|$ would imply $\Delta\energy_{i_k}>0$.
One can then extend this argument to a segment that needs to be flipped, as
illustrated in Figure~ \ref{fig:sketch-proof}.
\begin{figure}[h]
$$
\begin{array}{cc}
\includegraphics[width=3in]{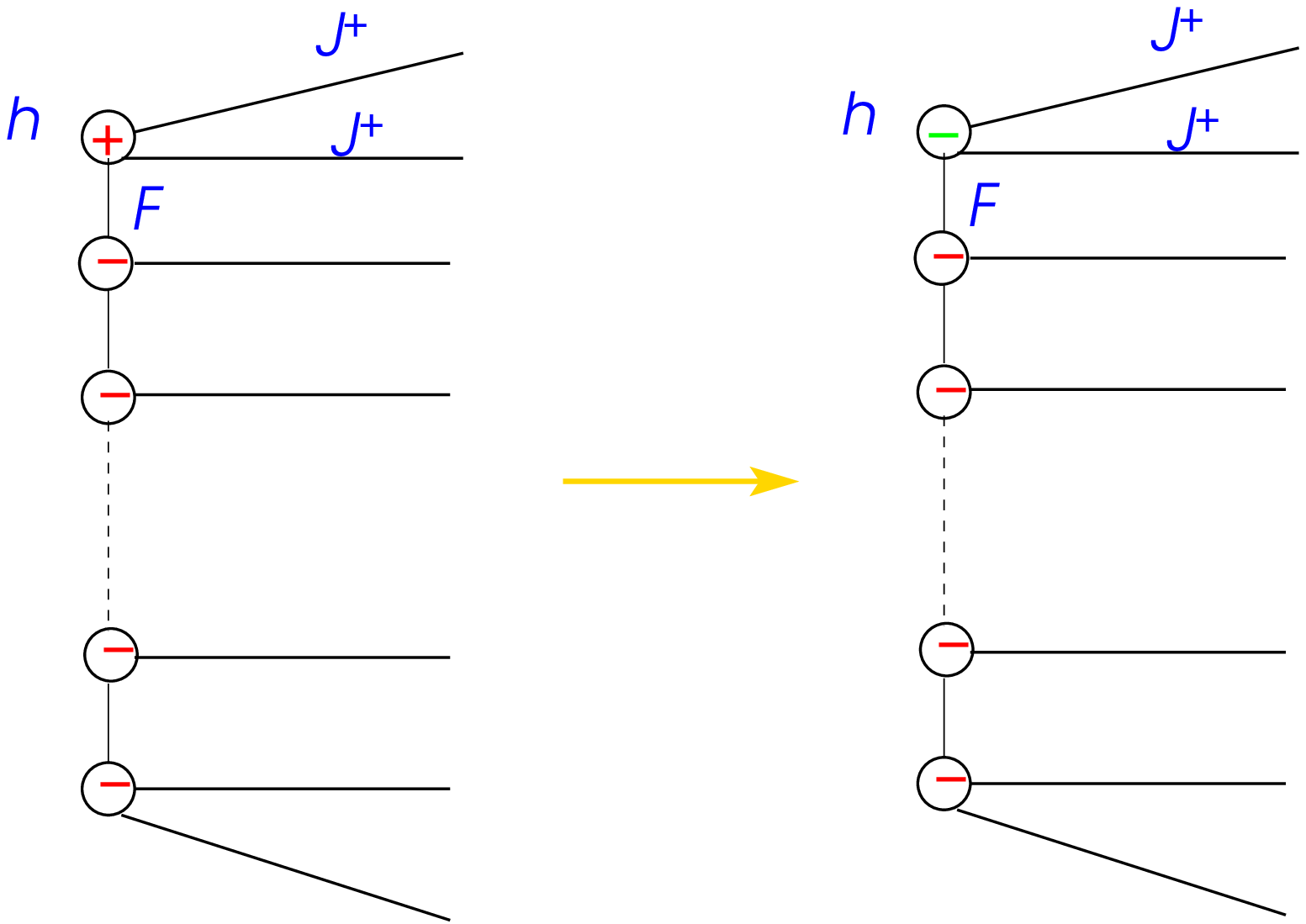} &
\includegraphics[width=3in]{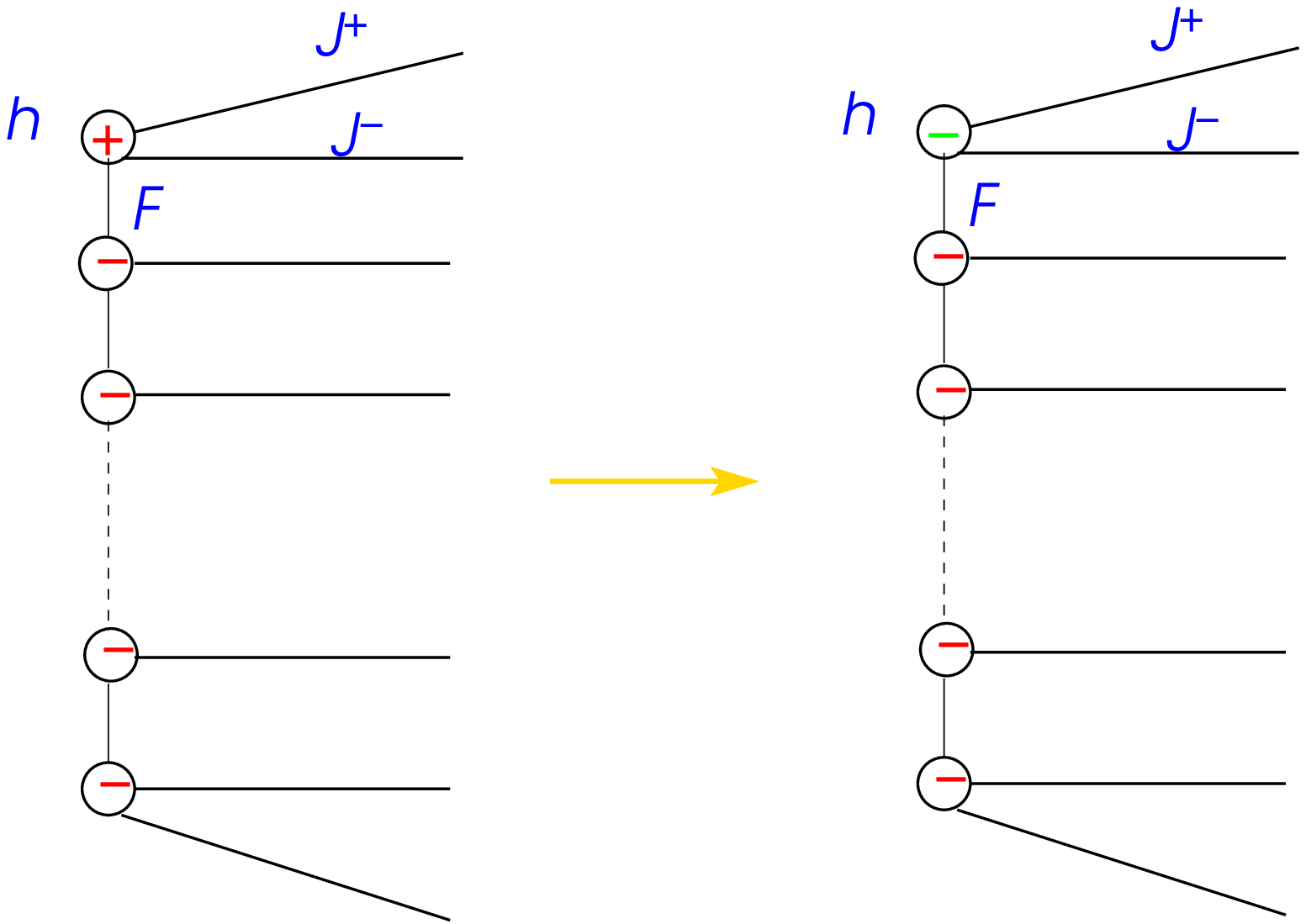}  \\
(a) h>0: \Delta\energy \ge 2 (h - J^+ - J^+ -F)  & (b) h>0: \Delta\energy \ge 2 (h - J^+ + J^- -F)  \\
\\
\includegraphics[width=3in]{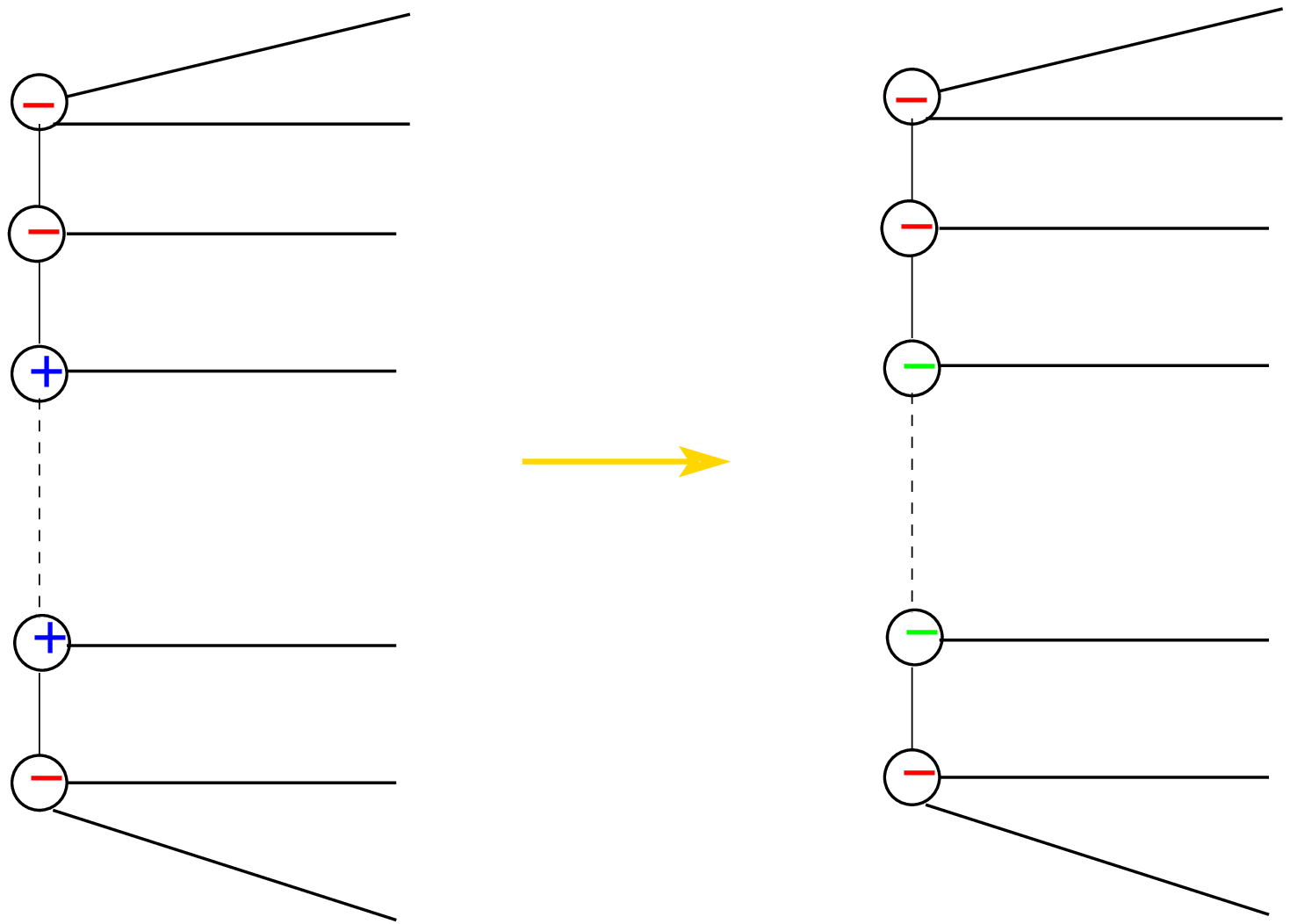} &
\includegraphics[width=3in]{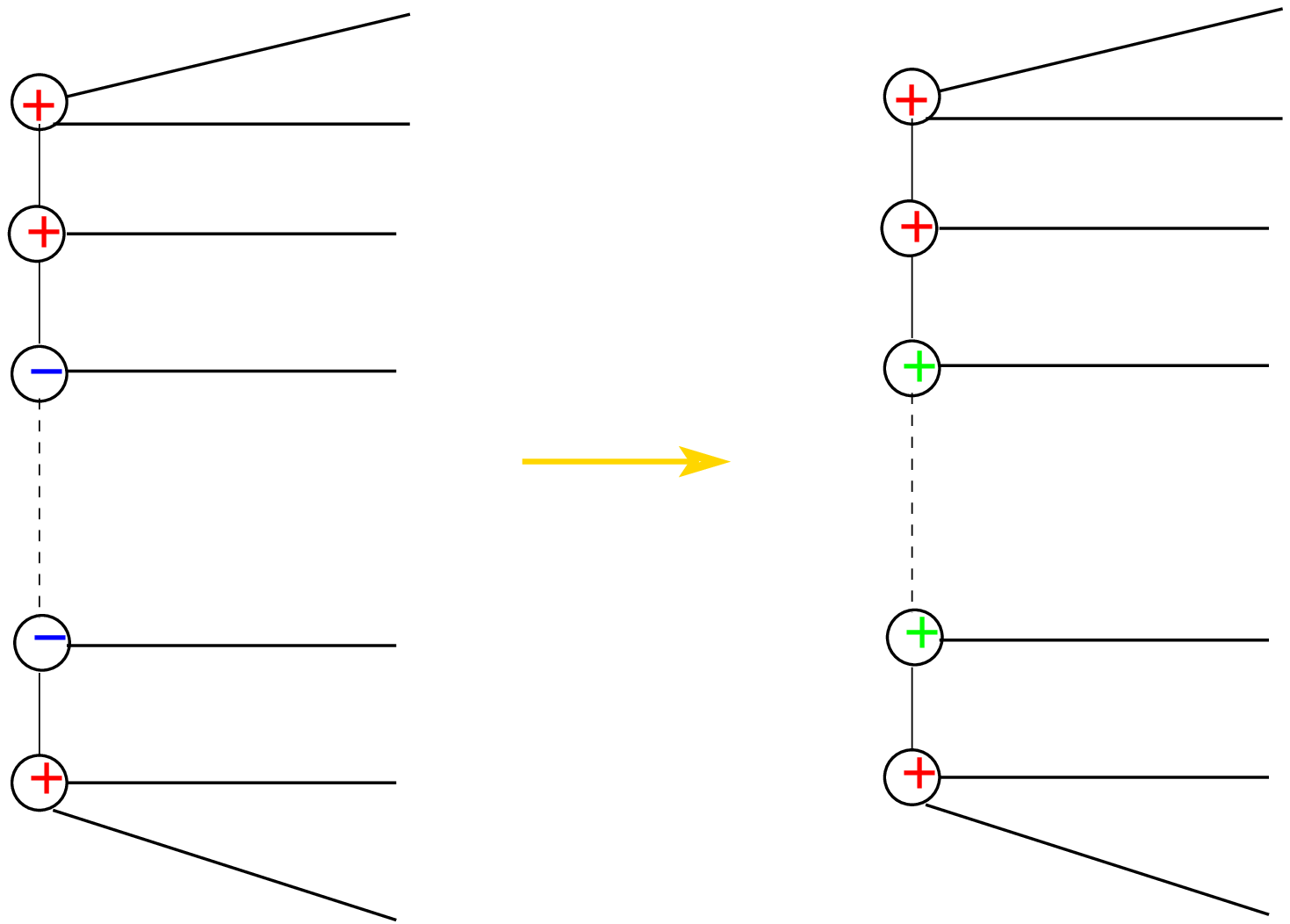}  \\
(c) h>0 & (d) h<0
\end{array}
$$
\caption{$(a)$ For $h>0$ and positive $J$s (denoted by $J^+$), the energy change $\Delta\energy \ge 2 (h - J^+ -
  J^+ -F)$ when flipping the leaf's spin from $+1$ to $-1$. 
$(b)$ For $h>0$, 
$\Delta\energy \ge 2 (h - J^+ + J^- -F)$ when flipping the leaf's spin from $+1$ to
  $-1$, where $J^+>0$, $J^-<0$.
(c) For $h>0$, flip a segment from $+1$ to $-1$. (d) For $h<0$, flip a segment from $-1$ to $+1$.}
\label{fig:sketch-proof}
\end{figure}

\begin{proof}
It is easy to check that $\sum_{i_k \in
  \ver(T_i)} h'_{i_k} = sign(h_i)|h_i| = h_i$.

We will first prove the theorem when $\Gemb$ is a 
topological-minor-embedding of $G$. 
That is, each $T_i$ is a chain $(i_1,i_2, \ldots,
i_{t_i})$ connected by consecutive vertices. 
Thus, the corresponding embedded energy
function is given by 
\begin{eqnarray}
    \energy^{\emb}(s_1, \ldots, s_N) 
    &=& \sum_{i \in \ver(G)} \left( \sum_{k=1}^{t_i} h'_{i_k} s_{i_k} + 
\sum_{k=1}^{t_i-1} F_i^{k(k+1)}s_{i_k}s_{i_{k+1}} \right) + \sum_{ij \in \edge(G)} J_{ij} 
s_{i_{\tau(i,j)}} s_{j_{\tau(j,i)}}
\label{eqn:Eemb-top}
\end{eqnarray}
In this case, we have $l_i=2$.
We prove by contradiction. 
Suppose NOT. That is, there exists $i \in \ver(G)$, 
and $1 \le k <t_i$  such that $s_{i_k}^*s_{i_{k+1}}^*=-1$.
We distinguish two cases based on the value of $h_i$.

For $h_i>0$,
let $p$ be the smallest index
such that $s_{i_{p-1}}^*=-1$ and $s_{i_p}^*=+1$, and let $q$ be the smallest index
such that  $s_{i_q}^*=+1$ and $s_{i_{q+1}}^*=-1$. That is, we have
$s_{i_1}^* = \ldots = s_{i_{p-1}}^* = -1$, $s_{i_p}^* = \ldots =
s_{i_q}^*=+1$. Note we have at least either $p>1$ or $q<n$. (Otherwise they
are all $+1$).
Then we claim that by flipping the segment $(i_p, \ldots, i_q)$ from $+1$s to $-1$s, the energy decreases, contradicting to the
optimality. 
(Notice that the $F$s within the segment do
 not change.) 
The energy change equals to
\begin{eqnarray*}
\Delta\energy^{\emb} &=& \energy^{\emb}(\ldots, -1,\ldots, -1, \underbrace{+1, \ldots, +1}, -1, \ldots )
 - 
\energy^{\emb}(\ldots, -1,\ldots, -1, \underbrace{-1, \ldots, -1}, -1, \ldots  )  
\\
&\ge& 2 \left(\sum_{k=p}^{q} h'_{i_k} - \sum_{k=p}^{q} \sum_{j_{\tau(j,i)} \in
   \onbr(i_k)}|J_{ij}| - F_i^{(p-1)p} - F_i^{q(q+1)} \right) \mbox{ (with convention that  $F_i^{01}=F_i^{t_i(t_i+1)}=0$)}
\\
&\ge&2 (-C_i/2 - F_i^{k(k+1)})~~\mbox{ (where $k=p$ if $p>1$, $k=q$ otherwise)} 
\end{eqnarray*} 

The argument for $h_i<0$ is similar except that we will flip from $-1$ to
$+1$ instead, and 
\begin{eqnarray*}
\Delta\energy^{\emb} &=& \energy^{\emb}(\ldots, +1,\ldots, +1, \underbrace{-1, \ldots, -1}, +1, \ldots )
 - 
\energy^{\emb}(\ldots, +1,\ldots, +1, \underbrace{+1, \ldots, +1}, +1, \ldots  )  
\\
&\ge& 2 \left(\sum_{k=p}^{q} -h'_{i_k} - \sum_{k=p}^{q} \sum_{j_{\tau(j,i)} \in
   \onbr(i_k)}|J_{ij}| - F_i^{(p-1)p} - F_i^{q(q+1)} \right) \\
&\ge&2 (-C_i/2 - F_i^{k(k+1)})~~\mbox{ (where $k=p$ if $p>1$, $k=q$ otherwise)} 
\end{eqnarray*}

Therefore, in both cases,
if 
$F_i^{p(p+1)} < -C_i/2$, for all $1 \le p<t_i$,
we have $\Delta\energy^{\emb} >0$, contradicting to the optimality. Hence all
$s_{i_p}^*$ must be of the same sign.

For the general minor-embedding when $T_i$s are trees, we can similarly
argue that for $h_i>0$, if all but one leaf is positive ($+1$), then
flipping them (from $+1$s to $-1$s) will decease the energy provided that
$F_i < -\frac{l_i-1}{l_i}C_i$. Similarly, one can argue for the case when $h_i<0$.

Notice that if we set $F_i= -\frac{l_i-1}{l_i}C_i - g_i/2$, then the
energy change $\Delta\energy^{\emb} \ge g_i$. In this case, the
spectral gap of the embedded Ising Hamiltonian will be the minimum
of $\min_{i \in \ver(G)} g_i$ and the spectral gap of the original Ising
Hamiltonian. 
\end{proof}

\paragraph{Remark.} One can generally set 
$ h'_{i_k} = sign(h_i) 
(  \sum_{j_{\tau(j,i)} \in \onbr(i_k)} |J_{ij}| - C_{i_k})
$ where $C_{i_k}$s are chosen such that $\sum_{i_k \in \ver(T_i)} C_{i_k}
= C_i (=\sum_{j \in \nbr(i)}|J_{ij}| - |h_i|)$, and set the $F$s
accordingly. This flexibility is useful when the precision for
parameters is limited. (In the above theorem, we give an upper bound for the
ferromagnetic coupler strengths. A natural question is how tight our bound
is. Can one get an better bound without solving the original problem?)




\section{Weighted Maximum Independent Set (WMIS) Problem}
\label{sec:wmis}
First, we formulate WMIS problem as a special case of QUBO in
Section~\ref{sec:mis-special-QUBO}. 
Then, we apply Theorem~\ref{thm:main} to set parameters for the
embedded Ising Hamiltonian
for the corresponding MIS problem in Section~\ref{sec:embed-mis}.

\subsection{Formulate WMIS Problem As a Special Case of QUBO}
\label{sec:mis-special-QUBO}
\paragraph{Weighted MIS (WMIS).}
Given an undirected vertex-weighted graph
$G$.
Let $\ver(G)=\{1, 2, \ldots, n\}$, let $c_i \in \reals^{+}$ be the weight of vertex
$i$. WMIS seeks to find a $\ms{wmis}(G)=S \subseteq V$ such that $S$ is independent
and the total weight of $S$ ($=\sum_{i \in S} c_i$) is maximized. 

\begin{theorem}
If $J_{ij} \ge \min\{c_i,c_j\}$ for all $ij \in \edge(G)$, then the maximum
  value of
$$ \oy(x_1,\ldots, x_n) = \sum_{i \in \ver(G)}c_i x_i - \sum_{ij \in \edge(G)}
  J_{ij}x_ix_j$$
is the total weighted of the WMIS. 
In particular if $J_{ij} > \min\{c_i,c_j\}$ for all
      $ij \in \edge(G)$, then $\ms{wmis}(G) = \{i \in \ver(G) : x^*_i =
1\}$,
where $(x^*_1, \ldots, x^*_n) = \argmax_{(x_1, \ldots, x_n) \in \{0,1\}^n}
\oy(x_1, \ldots, x_n)$.
\label{thm:mis}
\end{theorem}

\begin{proof}[{\bf Proof of Theorem~\ref{thm:mis}.}]
  Let $(x^*_1, \ldots, x^*_n) = \argmax_{(x_1, \ldots, x_n) \in \{0,1\}^n} \oy(x_1, \ldots, x_n)$.
Denote $S^* = \{i \in \ver(G): x^*_i=1\}$. 
We'll prove that if $J_{ij} > \min\{c_i,c_j\}$ for all $ij \in \edge(G)$,
then $S^*$ is an independent set.

    Suppose on the contrary, that is, there exists an edge $yz$
    in the subgraph
    induced by $S^*$. WLOG, assume that $c_y < c_z$. 
    Consider removing $y$ from
    $S^{*}$. Let $S' = S^{*} \setminus \{y\}$.
    The weight change  equals to $- c_y + \sum_{j \in \nbr(y) \cap S^{*}}J_{yj} \ge
    -c_y + J_{yz} >0$, contradicting to the
    optimality of $S^{*}$.
\end{proof}

The above theorem  is a generalization of the known fact for unweighted case
of MIS (see \cite{BH02} and references therein).
For the unweighted case  of MIS,  $c_i = 1$ for all $i \in \ver(G)$.
Thus, it is sufficient to choose $J_{ij} = 1+ \epsilon$ for all $ij \in
\edge(G)$ for  some $\epsilon>0$. 
Accordingly, the corresponding energy function of the Ising Model for MIS
is $\energy(s_1, \ldots, s_n) = \sum_{i \in \ver(G)} \left( \deg_i(1+\epsilon) - 2 \right)s_i - \sum_{ij \in \edge(G)} (1+\epsilon)s_is_j$.
\paragraph{Remark.}If we choose $J$ to be exactly $1$ instead (as in \cite{KL02,KLO04}), then we can only guarantee the
size of the maximum independent set, but the returned set is not
necessarily independent. 
For instance,  when $G=K_4$,  any (adjacent) two vertices also has the minimum
energy of $one$.

Note that we thus can conclude that the Ising problem is NP-hard because
WMIS is NP-hard. Indeed, Barahona ~\cite{Barahona82} showed the NP-hardness
of a special
Ising problem through the reduction of MIS problem on cubic
graph (which remains NP-hard). Notice that for WMIS on a planar graph, there is a PTAS
algorithm~\cite{Baker}. Recently, for the Ising problem on a planar graph, 
Bansal~et~al.~\cite{bansal-2007} applied the same technique
to obtain a PTAS algorithm in $O(n2^{36/\delta})$ where $(1-\delta)$ is the
approximation ratio. 
As another side note, conversely (to the fact that WMIS is a special case
of QUBO), Boros~et~al (see e.g. ~\cite{BH02}) has shown that a problem of QUBO
on $G$
can also be converted to a WMIS but on a different graph.

\subsection{Embedded Ising Hamiltonian for Solving MIS}
\label{sec:embed-mis}
Let $G$ be the graph of MIS problem and 
$\Gemb$ be the topological-minor-embedding of $G$ in $\uni$.
For the unweighted MIS, we have $C_i = \sum_{j \in \nbr(i)}J_{ij} - h_i =
2$. Therefore, according to Theorem~\ref{thm:main}, it suffices to set 
$F_i < -C_i/2 =-1$.
That is, the embedded Ising Hamiltonian for solving unweighted MIS is given
as in Eq.~\eqref{eqn:Eemb-top}, with 
$$ h'_{i_k} =  \left\{
\begin{array}{ll}
  d_{i_k} J - 1 & k=1, t_i\\
  d_{i_k} J & 1<k<t_i
\end{array}
\right.
$$ 
where $ d_{i_k} = |\onbr(i_k)|$,
and $J_{ij} \equiv J>1$, $F_{i}^{k(k+1)} \equiv F<-1$.
In particular, 
 for degree-3 hardware graph $\uni$, by setting
$J=1+\epsilon$, and $F = -(1+\epsilon)$, for some $\epsilon>0$, there are only 6 different
parameter values, namely, $\{-(1+\epsilon),0,\epsilon,1+\epsilon,
1+2\epsilon, 1+3\epsilon\}$ in Eq.~\eqref{eqn:Eemb-top}.

\section{Discussion}
\label{sec:discussion}
In this paper, we introduce minor-embedding in AQC. 
In particular, 
  we show that the NP-hard QUBO problem can be solved 
  using an adiabatic quantum computer that implements Ising spin-1/2
  Hamiltonians,  through minor-embedding reduction.
  There are two components to this reduction: embedding and 
    parameter setting. 
Given a minor-embedding, we show how to derive the values for the
corresponding parameters, in particular, a good upper bound (in terms of
magnitude) for the ferromagnetic coupler strengths, of the  embedded
  Ising Hamiltonian such that there is one-one correspondence between the
  groundstate of the original Ising Hamiltonian and the one of the embedded Ising
  Hamiltonian. 

There are many algorithmic
  problems related to minor-embedding in AQC that remain to be addressed. In
  particular, the problems relate to the efficiency of the
  reduction. These problems in turn relate to the running time or
  complexity of quantum
  adiabatic algorithms.
  Recall that according to the adiabatic theorem, the running time of an
  adiabatic algorithm depends on the
minimum spectral gap of the system Hamiltonian, which however might be as
hard as solving the original problem. 
Despite several serious investigations, 
the power of AQC remains an open
question~\cite{DMV01,DV01,Reichardt-04,ioannou-2007,ADKLLR04}.
How does the embedding reduction effect the time complexity of an adiabatic
algorithm?
In order to address this question, one will also need to specify the initial
Hamiltonian. In~\cite{AC08}, we show that for some special cases, 
how the embeddings,
parameters, and initial Hamiltonians can effect the minimum spectral gaps.
The effect of the embedding  and its consequential initial
Hamiltonian
on the complexity of adiabatic algorithms remains to be
investigated.
In the following we state several main problems that need to investigate in
order to facilitate the design of adiabatic algorithms and
AQC architectures. Partial
results to these problems will appear in our subsequent papers.

\paragraph{P1. Measurement for the minor-embedding.} Define a measure for the
minor-embedding such that a good minor-embedding corresponds to a reduced
problem that admits an efficient adiabatic algorithm. 

\paragraph{P2. Embedding-dependent initial Hamiltonian.} Design an
embedding-dependent initial Hamiltonian for a given minor-embedding such that the adiabatic algorithm
for the reduced problem is at
least as efficient as the adiabatic algorithm (with the best possible initial
Hamiltonian) for the original problem.

\paragraph{P3. Hardware graph design.} 
Given a family $\mathcal{F}$ of graphs (which consists of classically hard
instances), the problem is to design a hardware graph (called a {\em
  $\mathcal{F}$-minor-universal} graph) which
is as small as possible (in terms of total number of vertices and edges)
such that 
\begin{itemize}
\item all known physical constraints are satisfied;
\item all graphs in $\mathcal{F}$ are embeddable;
\item a good embedding of each graph in $\mathcal{F}$ can be efficiently
  computed. 
\end{itemize}

\vnote{treewidth, and Alon's work}

\section*{Acknowledgment}
I would like to thank my incredible colleagues:   Mohammad Amin,
Andrew Berkley,      Richard Harris, Mark Johnson,
      Jan Johannson,
      Andy Wan,       
      Colin Truncik, 
      Paul Bunyk,
      Felix Maibaum,
Fabian Chudak,
      Bill Macready,
       and Geordie Rose.
Thanks also go to David Kirkpatrick for the discussion, advice and  encouragement.
I would also like to thank Bill Kaminsky for his detailed comments.

\begin{thebibliography}{10}

\bibitem{ADKLLR04}
D.~Aharonov, W.~van Dam, J.~Kempe, Z.~Landau, S.~Lloyd, and O.~Regev.
\newblock Adiabatic quantum computation is equaivalent to standard quantum
  computation.
\newblock {\em Proc. 45th FOCS}, pages 42--51, 2004.

\bibitem{AC08}
M.~H.~S. Amin and V.~Choi.
\newblock Work in progress.
\newblock 2008.

\bibitem{amin-2008-100}
M.~H.~S. Amin, P.~J. Love, and C.~J.~S. Truncik.
\newblock Thermally assisted adiabatic quantum computation.
\newblock {\em Physical Review Letters}, 100:060503, 2008.

\bibitem{amin-2008}
M.~H.~S. Amin, C.~J.~S. Truncik, and D.~V. Averin.
\newblock The role of single qubit decoherence time in adiabatic quantum
  computation.
\newblock {\em arXiv.org:0803.1196}.

\bibitem{Baker}
B.~S. Baker.
\newblock Approximation algorithms for {NP}-complete problems on planar graphs.
\newblock {\em J. ACM}, 41(1):153--180, 1994.

\bibitem{bansal-2007}
N.~Bansal, S.~Bravyi, and B.~M. Terhal.
\newblock A classical approximation scheme for the ground-state energy of
  {I}sing spin hamiltonians on planar graphs.
\newblock {\em quant-ph/0705.1115}.

\bibitem{Barahona82}
F.~Barahona.
\newblock On the computational complexity of {I}sing spin glass models.
\newblock {\em J. Phys. A: Math. Gen.}, pages 15: 3241--3253, 1982.

\bibitem{BH02}
E.~Boros and P.~Hammer.
\newblock Pseudo-boolean optimization.
\newblock {\em Discrete Appl. Math.}, (123):155--225, 2002.

\bibitem{BHT06}
E.~Boros, P.~L. Hammer, and G.~Tavares.
\newblock Preprocessing of quadratic unconstrained binary optimization.
\newblock {\em Technical Report RRR 10-2006, RUTCOR Research Report.}, 2006.

\bibitem{bravyi-2008}
S.~Bravyi, D.~P. DiVincenzo, D.~Loss, and B.~M. Terhal.
\newblock Simulation of many-body hamiltonians using perturbation theory with
  bounded-strength interactions.
\newblock {\em arXiv.org:0803.2686}.

\bibitem{CFP02}
A.~Childs, E.~Farhi, and J.~Preskill.
\newblock Robustness of adiabatic quantum computation.
\newblock {\em Physical Review A}, 65, 2002.

\bibitem{Distel05}
R.~Diestel.
\newblock {\em Graph Theory}.
\newblock Springer-Verlag, Heidelberg, 2005.

\bibitem{Leighton-book}
L.~F.~Thomson.
\newblock {\em Introduction to parallel algorithms and architectures: arrays,
  trees, hypercubes.}
\newblock 1992.

\bibitem{FGGLLP01}
E.~Farhi, J.~Goldstone, S.~Gutmann, J.~Lapan, A.~Lundgren, and D.~Preda.
\newblock A quantum adiabatic evolution algorithm applied to random instances
  of an np-complete problem.
\newblock {\em Science}, 292(5516):472--476, 2001.

\bibitem{FGGS00}
E.~Farhi, J.~Goldstone, S.~Gutmann, and M.~Sipser.
\newblock Quantum computation by adiabatic evolution.
\newblock {\em quant-ph/0001106}, 2000.

\bibitem{ioannou-2007}
L.~M. Ioannou and M.~Mosca.
\newblock Limitations of some simple adiabatic quantum algorithms.
\newblock {\em quant-ph/0702241}, 2007.

\bibitem{KL02}
W.~M. Kaminsky and S.~Lloyd.
\newblock Scalable architecture for adiabatic quantum computing of {NP}-hard
  problems.
\newblock In A.~J. Leggett, B.~Ruggiero, and P.~Silvestrini, editors, {\em
  Quantum Computing and Quantum Bits in Mesoscopic Systems}, 2004.

\bibitem{KLO04}
W.~M. Kaminsky, S.~Lloyd, and T.~P. Orlando.
\newblock Scalable superconducting architecture for adiabatic quantum
  computation.
\newblock {\em arXiv.org:quant-ph/0403090}, 2004.

\bibitem{kempe-2006-35}
J.~Kempe, A.~Kitaev, and O.~Regev.
\newblock The complexity of the local hamiltonian problem.
\newblock {\em SIAM JOURNAL OF COMPUTING}, 35:1070, 2006.

\bibitem{KR96}
J.~M. Kleinberg and R.~Rubinfeld.
\newblock Short paths in expander graphs.
\newblock In {\em {IEEE} Symposium on Foundations of Computer Science}, pages
  86--95, 1996.

\bibitem{oliveira-2005}
R.~Oliveira and B.~M. Terhal.
\newblock The complexity of quantum spin systems on a two-dimensional square
  lattice.
\newblock {\em quant-ph/0504050}.

\bibitem{Reichardt-04}
B.~W. Reichardt.
\newblock The quantum adiabatic optimization algorithm and local minima.
\newblock In {\em STOC '04: Proceedings of the thirty-sixth annual ACM
  symposium on Theory of computing}, pages 502--510, New York, NY, USA, 2004.
  ACM.

\bibitem{RS95}
N.~Robertson and P.~D. Seymour.
\newblock Graph minors. xiii: the disjoint paths problem.
\newblock {\em J. Comb. Theory Ser. B}, 63(1):65--110, 1995.

\bibitem{DMV01}
W.~van Dam, M.~Mosca, and U.~Vazirani.
\newblock How powerful is adiabatic quantum computation?
\newblock {\em Proc. 42nd FOCS}, pages 279--287, 2001.

\bibitem{DV01}
W.~van Dam and U.~Vazirani.
\newblock Limits on quantum adiabatic optimization.
\newblock {\em Unpublished}, 2001.

\end{thebibliography}
\bibliographystyle{abbrv}

\newpage
\appendix
\centerline{\bf \large Appendix}


\section{Energy Function of the Ising Hamiltonian}
In this section, we show that the energy function of the Ising Hamiltonian
in Eq.~\eqref{eqn:final-ham} is given in Eq.~\eqref{eqn:OE}.
That is, replace each $\sigma^z_i$ by the
  variable $s_i \in \{-1,+1\}$. 
The equivalence follows from two basics: (1) Eigenvalues and eigenvectors of
$\sigma^z$: $\sigma^z\ket{0} = (+1) \ket{0}$, $\sigma^z\ket{1} = (-1)
\ket{1}$,
where $\ket{0} \mdef [1,0]^{\dagger}$, $\ket{1} \mdef [0,1]^{\dagger}$.
(2) Tensor product property: $(A \otimes B)(\ket{u} \otimes \ket{v}) =
A\ket{u} \otimes B\ket{v}$.
More precisely, for $\ket{z} = \ket{z_1} \otimes \ldots \otimes \ket{z_n}$,
$z_i \in \{0,1\}$, $i=1, \ldots, n$.
  \begin{eqnarray*}
 \ham_{\ms{final}} \ket{z}  &=& \left(\sum_{i \in \ver(G)} h_i \sigma^z_i + \sum_{ij \in \edge(G)} J_{ij}
\sigma^z_i \sigma^z_j \right)  \ket{z_1} \otimes \ldots \otimes \ket{z_n}\\
&=& \sum_{i \in \ver(G)} h_i \ket{z_1} \otimes \ldots \otimes 
(-1)^{z_i}\ket{z_i} \ldots \otimes \ket{z_n}  + \sum_{ij \in \edge(G)} J_{ij} 
\ket{z_1} \otimes \ldots \otimes 
(-1)^{z_i}\ket{z_i} \ldots \otimes (-1)^{z_j}\ket{z_j} \ldots \otimes \ket{z_n}\\
&=& \left(\sum_{i \in \ver(G)} h_i (-1)^{z_i} + \sum_{ij \in \edge(G)} J_{ij}
(-1)^{z_i} (-1)^{z_j}\right)  \ket{z_1} \otimes \ldots \otimes \ket{z_n}
  \end{eqnarray*}

Therefore, the energy function of $\ham_{\ms{final}}$ is 
$$ \energy (z_1, \ldots, z_n) =  \sum_{i \in \ver(G)} h_i (-1)^{z_i} + \sum_{ij \in \edge(G)} J_{ij}
(-1)^{z_i} (-1)^{z_j}$$ where $z_i \in \{0,1\}$, for $i=1, \ldots, n$.
Replace $(-1)^{z_i}$ by $s_i \in \{+1,-1\}$, we thus have 
Eq.\eqref{eqn:OE}.

\end{document}